\documentclass[sigconf,nonacm]{acmart}

\usepackage{caption}
\DeclareCaptionType{copyrightbox}
\usepackage{rotating}
\usepackage[flushleft]{threeparttable}
\usepackage{placeins}
\usepackage[all]{xy}
\usepackage{xcolor}
\usepackage{subcaption}
\usepackage{amsmath, bm}

\usepackage{xcolor}
\usepackage{colortbl}

\usepackage{algorithm}
\usepackage[noend]{algorithmic}
\usepackage{comment}

\usepackage{amsthm}

\newtheorem{lemma}{Lemma}[section]

\newcommand{\chase}{\textsc{Chase}}
\newcommand{\cbchase}{\textsc{Chase}\&\textsc{Backchase}}
\newcommand{\bchase}{\textsc{Backchase}}

\newcommand{\where}{\textbf{\textit{where}}}

\newcommand{\why}{\textbf{\textit{why}}}

\newcommand{\how}{\textbf{\textit{how}}}
\newcommand{\How}{\textbf{\textit{How}}}

\usepackage{cancel}
\definecolor{Red}{rgb}{1,0,0}
\newcommand{\colorcancel}[2]{\renewcommand{\CancelColor}{\color{#2}}\cancel{#1}} 

\begin{document}

\title{Enhanced Inversion of Schema Evolution with Provenance}

\author{Tanja Auge}
\email{tanja.auge@uni-rostock.de}
\affiliation{%
\institution{University of Rostock}
\city{Rostock}
\country{Germany}}

\author{Andreas Heuer}
\email{andreas.heuer@uni-rostock.de}
\affiliation{%
\institution{University of Rostock}
\city{Rostock}
\country{Germany}}

\renewcommand{\shortauthors}{Auge and Heuer}

\pagestyle{empty}

\begin{abstract}
Long-term data-driven studies have become indispensable in many areas of science. Often, the data formats, structures and semantics of data change over time, the data sets evolve. Therefore, studies over several decades in particular have to consider changing database schemas. The evolution of these databases lead at some point to a large number of schemas, which have to be stored and managed, costly and time-consuming. However, in the sense of reproducibility of research data each database version must be reconstructable with little effort. So a previously published result can be validated and reproduced at any time. 

Nevertheless, in many cases, such an evolution can not be fully reconstructed. This article classifies the 15 most frequently used schema modification operators and defines the associated inverses for each operation. For avoiding an information loss, it furthermore defines which additional provenance information have to be stored. We define four classes dealing with dangling tuples, duplicates and provenance-invariant operators. Each class will be presented by one representative. 

By using and extending the theory of schema mappings and their inverses for queries, data analysis, \why-provenance, and schema evolution, we are able to combine data analysis applications with provenance under evolving database structures, in order to enable the reproducibility of scientific results over longer periods of time. While most of the inverses of schema mappings used for analysis or evolution are not exact, but only quasi-inverses, adding provenance information enables us to reconstruct a sub-database of research data that is sufficient to guarantee  reproducibility. 
\end{abstract}


\maketitle

\section{Introduction}
Research institutions all around the world produce huge amounts of research data, which have to be managed by \textit{research data management}. This includes collecting, evaluating, analyzing, archiving, and publishing the original or processed research data as well as the results of (data) analyses, theoretical considerations, or other scientific investigations. In this paper we are focusing on structured data, represented as a relational database, generated in scientific experiments, measurement series or always-on sensors.

The presentation and publication of research results increasingly requires publishing the corresponding research data, which ensures the findability, accessibility, inter-operability, and re-usability in the sense of \textit{FAIR-Principles}\footnote{\url{https://www.go-fair.org/}}. FAIR does not require the publication of all source data. It is sufficient to specify the source tuples necessary for the calculation of a specific research result. To perform this calculation, we use the so-called \textit{\chase\ algorithm}  supplemented with provenance information such as witness bases \cite{BKT01}, provenance polynomials \cite{GKT07} and/or additional side tables.

The \chase\ is a procedure that modifies a database instance $I$ by incorporating a set of dependencies $\Sigma$ represented as (s-t) tgds or egds. Originally developed for database design and semantic query optimization, the \chase\ nowadays is also used for database applications such as data exchange, data integration or answering queries using views. Formulated as schema mapping $\mathcal{M} = (I,J,Q)$ using s-t tgds, i.e. a dependency from one database $I$ into a second database $J$, the \chase\ can process database queries $Q$. This is equivalent to a mapping from red to green in Figure \ref{fig:chase}. Thus, the \chase\ can be used for evaluating conjunctive queries or simple SQL queries, such as \texttt{SELECT n,m FROM F,S WHERE F.m=S.m}. We call these queries against the research data \textit{evaluation queries}.

\begin{figure}[hb]
\centering
\includegraphics[width=0.75\linewidth]{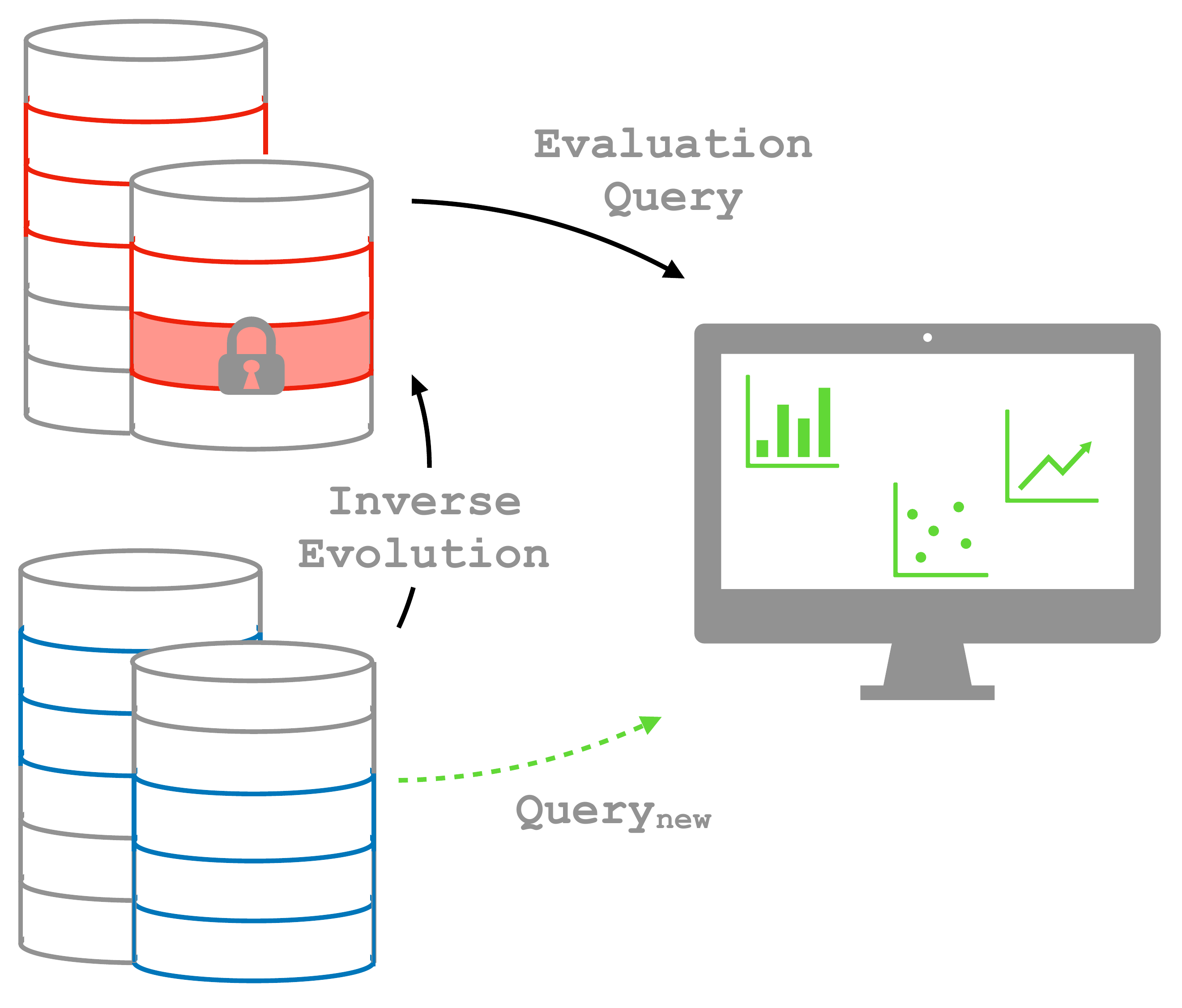}
\caption{Chasing evaluation queries and evolution mappings}
\label{fig:chase}
\end{figure}

In addition, the \chase\ can be used for inverting queries. For this, we choose $J$ as database instance and the inverse evaluation query $Q^{-1}$ as schema mapping $\mathcal{M}^\ast = (J,I,Q^{-1})$ formalized as s-t tgd like before. Running both evaluation query and inversion in sequence, this results in a \cbchase\ process \cite{DPT99}. Here the \bchase\ phase corresponds to a second \chase\ phase with the \chase\ result of the first phase as input. However, the result $I^\ast=\chase_{\mathcal{M}^\ast}(\chase_\mathcal{M}(I))$ after applying \cbchase\ to $I$, $I^\ast$ may not be equal to $I$. In general, $I^\ast$ is a subset of $I$, called \textit{(minimal) sub-database} (highlighted in red, Figure \ref{fig:chase}). The corresponding inverse, calculated with the \chase\, is called \textit{\chase-inverse}. Such inverses are not necessarily exact and unique. They are so-called \textit{pseudo-inverse functions} of different granularity \cite{FKPT11,AH18b}.

In research data management, not only the data is changing (e.g. by adding results of new experiments), but also the data's schemas. Changing schemas often results in freezing the entire research database to be able to perform the evaluation queries used so far. However, this can be time-consuming, expensive and leads to a dramatically increasing amount of data to be archived. Formulating a schema evolution $E$ as a schema mapping by s-t tgds allows processing it by the \chase\ just like the evaluation queries before. The composition of inverse evolution steps and original evaluation queries now enables the reconstruction of research data for changing database structures, summarized in Figure \ref{fig:chase}. 

At first glance, choosing the \chase\ algorithm may seem a bit unusual, but the choice has some significant benefits:
\begin{enumerate}
\item Both schema evolution and simple evaluation queries can be interpreted as schema mappings and processed with the \chase\ algorithm \cite{FKPT11}. So, both mappings can be treated with the same theory.
\item An elaborated theory about quasi-inverses of schema mappings as well as their composition already exists \cite{FKPT08,FKPT11}.
\item An examination of different types of inverses of the evaluation queries can be found in \cite{AH18b}. It can be improved by additional provenance information. This allows the reconstruction of lost null tuples and duplicates, as well as the replacement of null values by concrete values.
\end{enumerate}
Interpreting evolution and evaluation as schema mapping and processing by \cbchase\ enables the reconstruction of a research result even in the presence of changing database structures (see Figure \ref{fig:chase}). For this, an evolution mapping, represented as \textit{schema modification operator}, an evaluation query as well as the two inverses are formulated as s-t tgds. The \chase\ calculates the query result, which is then processed by the \bchase. This provides the (minimal) sub-database in question (highlighted in \textcolor{red}{red}). 

The inverse type can be used to determine how much and what part of the original data can be recovered without additional effort. In \cite{AH18b} the \chase-inverse types of 14 typical evaluation operators were specified with and without using data provenance. Provenance deals with the traceability of a result back to its original, possibly physical object. In the case of data provenance, we are interested in where a result tuple comes from, why and how it is calculated. 
In \cite{AH18b}, it is observed that provenance improves the inverse types in about half of the cases. 

In this article, we will perform this analysis for the evolution. For this purpose, we are looking at 15 schema modification operators, which are particularly common in practice \cite{CMZ08,CMTZ08,AMJFH20,GZ12,WN11,QLS13}, and verify their invertibility. Then we improve the inverse type again by adding provenance. This again improves about half of the cases. Other approaches for combining schema evolution and provenance are given in \cite{GAMH10} and \cite{GZ12}. In total, we identify four classes of invertibility. For each class we consider one representative in more detail (see Table \ref{tab:SMOInverseExtract}).

\paragraph{\textbf{Structure of the Article}}
To determine the relevant original tuples, we use the \chase\ algorithm not only for processing an evaluation query but also for evolution, see Section \ref{sec:inverses}. For this purpose, in Section \ref{sec:Examples}, we divide schema modification operators into four classes and examine five representatives with respect to their \chase-inverses with and without additional provenance information. Before that, in Section \ref{sec:SotA}, we introduce the necessary definitions regarding \chase\ and data provenance.

\section{Basics and Related Work}
\label{sec:SotA}

This article combines \chase, inverse schema mappings and data provenance from different areas of database theory. We introduce the most important terms here, further details can be found in the respective literature.

\paragraph{\textbf{\chase}}

The \textit{\chase\ algorithm} is a technique used in a number of data management tasks like data exchange, data cleaning or answering queries using views. Its wide range of applicability can be attributed to the fact that \chase\ generalizes well. Generally speaking, it incorporates a set of dependencies $\ast$ into a given object $\bigcirc$ -- usually an instance or a query. Only condition: The dependencies must be representable as (s-t) tgds and egds.

A \textit{source-to-target tuple generating dependency} (s-t tgd) is a formula of the form $\forall x : (\phi(x) \rightarrow \exists y : \psi(x, y))$, with $x$,$y$ tuples of variables, and $\phi(x)$, $\psi(x,y)$ are conjunctions of atoms, called \textit{body} and \textit{head}. It can be seen as inter-database dependency, representing a constraint within a database. An intra-database dependency is called \textit{tuple generating dependency} (tgd). A dependency with equality atoms in the head \textit{equality generating dependency} (egd). 

In the case of instances, chasing (s-t) tgds create new tuples and chasing egds clean the database by replacing null values (by other null values or constants). So, the \chase d database satisfies all dependencies of $\ast$. 
We use the \chase\ to process schema evolution. Thus, we deal with inter-database mapping and can restrict ourselves to s-t tgds in the following. Since chasing s-t tgds always terminates, we refer to \cite{BKT17,GMS12} for further explanations about \chase\ variants and termination criteria.

\setlength{\tabcolsep}{0.5em} 
{\renewcommand{\arraystretch}{1.2}
\begin{table*}[ht]
\caption{Sufficient and necessary conditions of \chase-inverses, based on \cite{FKPT11,AH18b}}
\label{tab:CHASEinverse}
\center
\small
\scalebox{0.95}{
\begin{tabular}{|l|c|l|}
\hline
\multicolumn{1}{|l|}{\textbf{\chase-inverse}} & \textbf{sufficient condition} & \textbf{necessary condition} \\
\hline \hline
exact ($=$) & $I = I^\ast$  & ~$I^\ast = I$ \\
\hline 
classical ($\equiv$) & $=$ & $\exists$ homomorphism $h$ : $I^\ast \rightarrow I$, $\exists$ homomorphism $h'$ : $I \rightarrow I^\ast$ \\ 
\hline 
tp-relaxed $(\preccurlyeq_\textrm{tp})$ & $=$ & $\exists$ homomorphism $h$ : $I^\ast \rightarrow I$, $\mid I^\ast \mid = \mid I \mid$, $I^\ast \leftrightarrow_\mathcal{M} I$ \\  
\hline 
relaxed ($\preccurlyeq$) & $\preccurlyeq_\textrm{tp}$ or $\equiv$ & $\exists$ homomorphism $h$: $I^\ast \rightarrow I$, $I^\ast \leftrightarrow_\mathcal{M} I$ \\
\hline 
result equivalent ($\leftrightarrow$) & $\preccurlyeq$ & $I^\ast \leftrightarrow_\mathcal{M} I$ \\ 
\hline 
\multicolumn{3}{c}{}
\end{tabular}} 
\end{table*}}

\paragraph{\textbf{\chase-inverse functions}}

Let $I$ and $J$ be two database instances before and after applying the \chase, i.e. $J = \chase(I)$. A \textit{\chase-inverse} is a function from $J$ to $I^\ast$, which applies the \chase\ algorithm to $J$ with $I^\ast$ section of $I$. $I^\ast$ is a \textit{section} of $I$ (short: $I^\ast \preccurlyeq I$) iff a homomorphism $h:I^\ast \rightarrow I$ exists, whereby constants are mapped to constants and null values are mapped to constants or (other) null values \cite{AH18b}. \chase-inverse functions belong to the quasi-inverse functions. Formally, we call $g: I^\ast \rightarrow I$ a \textit{quasi-inverse function} to $f:I \rightarrow J$ iff $I^\ast \preccurlyeq I$ and $f \circ g \circ f = f$ \cite{FKPT08}. This means that quasi-inverse functions usually recover the original data only partially. 
In the case of \chase\ five inverse types are distinguish: exact, classical, relaxed \cite{FKPT11}, tuple-preserving relaxed and result equivalent \cite{AH18b}. The sufficient and necessary conditions of the different \chase-inverse types can be seen in Table \ref{tab:CHASEinverse}. 

For this, let $\mathcal{M}$ be a schema mapping, $I$ and $I^\ast$ two database instances with $I^\ast = \chase_{\mathcal{M}^\ast}(\chase_\mathcal{M}(I)).$ $I^\ast$ is called \textit{sub-database} of $I$.
A \textit{schema mapping} $\mathcal{M}$ can be formalized as a triple $(S_t,S_{t+1},\Sigma)$ with $S_t$ source schema, $S_{t+1}$ target schema and $\Sigma$ set of dependencies that specify the relationship between $S_{t}$ and $S_{t+1}$. The affiliated inverse mapping is defined as $\mathcal{M}^\ast = (S_{t+1} , S_t , \Sigma^{-1})$. The star $^\ast$ symbolizes that the inverse function $\mathcal{M}^\ast$ may not be exact or unique. Again it is a quasi-inverse function.

While an \textit{exact \chase-inverse} ($=$) reconstructs the original database instance $I$, i.e. $I^\ast = I$, the \textit{classical \chase-inverse} ($\equiv$) only returns an instance $I^\ast$, which is equivalent to the original database instance $I$, i.e. $I^\ast$ and $I$ are equal except for isomorphism. The \textit{relaxed \chase-inverse} ($\preccurlyeq$) does not require an equivalent relationship between the original instance $I$ and the recovered instance $I^\ast$, but data exchange equivalence (short: ${I \leftrightarrow_\mathcal{M} I^\ast}$) and the existence of a homomorphism between $I^\ast$ and $I$ \cite{FKPT08}.

To preserve the number of tuples, the definition of relaxed \chase-inverse is extended \cite{AH18b}. The corresponding function, called \textit{tp-relaxed \chase-inverse} ($\preccurlyeq_\textrm{tp}$), additionally requires $\mid\! I^\ast\!\mid = \mid\!I\!\mid$. In contrast, a weakened definition is the \textit{result equivalent \chase-inverse} ($\leftrightarrow$), which only satisfies data exchange equivalence, i.e. $\chase_\mathcal{M}(I) \equiv \chase_\mathcal{M}(I^\ast)$.

As studied in \cite{AH18b}, evaluation queries formalised as s-t tgds and applied by \chase\ are classified as exact, (tp-)relaxed or result equivalent depending on query under consideration. The addition of provenance information like provenance polynomials \cite{GT17} and (minimal) witness basis \cite{BKT01} allows the specification of stronger \chase-inverses than without. For example, in the case of projection or union, a stronger \chase-inverse can be constructed by using provenance. For other operations, such as selection, the inverse type can not be improved despite additional information.

\paragraph{\textbf{\chase\&\bchase}}
A schema mapping can be processed using the \chase\ and be inverted using the \bchase\ \cite{DPT99}. The \chase\ phase computes a modified instance $J$ generated from $I$ by applying the \chase\ using a schema mapping $\mathcal{M} = (S_t,S_{t+1},\Sigma)$. After that, the \bchase\ phase computes an instance $I^\ast$ by applying the \chase\ to $J$ using a schema mapping $\mathcal{M}^\ast = (S_{t+1},S_t,\Sigma^{-1})$:
$$I^\ast = \underbrace{\chase_{\mathcal{M}^\ast}(J)}_\bchase=\chase_{\mathcal{M}^\ast}(\underbrace{\chase_\mathcal{M}(I)}_\chase).$$
The calculated instance $I^\ast$ can, but must not be sub-database of $I$, short $I^\ast \preccurlyeq I$. This depends on the inverse type defined above. 

\paragraph{\textbf{Data Provenance}}
Depending on the authors (see \cite{HDL17} or \cite{PRSA18}), a distinction is made between different types of provenance. These include data provenance, workflow provenance and metadata provenance. The questions to be answered depend on the provenance type and can be answered in different ways. In this article, we will focus on data provenance. 

Let be $I$ database instance and $Q$ query. \textit{Data pro\-venance} describes (1) where a result tuple $r \in Q(I)$ does come from, (2) why and (3) how $r$ exists in the result $Q(I)$. This can be answered extensively by specifying concrete tuples and database values, intensionally by characterizing a set of database values \cite{Mot89} or query-based by query transformation \cite{BHT14}. We consider only extensional answers. 

The so-called \why-provenance \cite{BKT01} specifies a witness basis that identifies the tuples involved in the calculation of $r$. This tuple-based basis contains all ids necessary for reconstructing $r$. \How-provenance, on the other hand, uses provenance polynomials \cite{GT17,Sel17} to explain how a result tuple $r$ is calculated. The polynomials are defined by a commutative semi-ring $(\mathbb{N}\lbrack X\rbrack,+,\cdot,0,1)$ with $+$ for union and projection, as well as $\cdot$ for natural join \cite{GKT07}. 
If we are simply interested in identifying ID or relation name of the tuples involved, we choose \where-provenance. It is the weakest of the three provenance questions, but is already sufficient in many cases.

\paragraph{\textbf{Schema Evolution}}
Schema evolution deals with the problem of maintaining schema and functionality of a software system considering a changing data structure. It can be addressed by so-called \textit{Schema Modification Operators} (SMOs) that support semantic conversion of columns, column concatenation/splitting and other schema modifications. Common operators in the case of research data management include \texttt{ADD Column} and \texttt{MERGE Column} as well as \texttt{SPLIT Column} \cite{AMJFH20}. Other applications are described in \cite{QLS13,WN11}. For initial approaches to combining schema evolution and provenance, see \cite{GAMH10} and \cite{CMZ08}.

\section{SMOs and their \chase-inverses}
\label{sec:inverses}

Especially in the case of long-term data-driven studies, the specification of source data may become a problem. Therefore, we process not only the evaluation but also the evolution by the \chase\&\textsc{Back-chase} and extend it with data provenance. For this, we classify the 15 most frequently used schema modification operators (SMOs) and divide them into four classes. We define the associated inverses for each operation and examine it with additional provenance information. While the first class includes all operators with an exact inverse, i.e.\,the operators are provenance-invariant, the second and third class deal with dangling tuples and duplicates. The remaining operators we summarize in Class IV.

\paragraph{\textbf{\cbchase}}
Evaluation queries and schema evolution can be defined as schema mappings. Let be $S_t$, $S_{t+1}$ two database schemas and $\Sigma$, $\Sigma^{-1}$ two sets of dependencies, whereby $\Sigma$ corresponds to an evaluation query $Q$ and $\Sigma^{-1}$ corresponds to an schema modification operator. Then the evolution $E$ can be formalized as schema mapping $\mathcal{M} = (S_t,S_{t+1},\Sigma)$ and its corresponding inverse $E^{-1}$ as $\mathcal{M}^\ast = (S_{t+1},S_t,\Sigma^{-1})$. Let further be $I(S_t)$ and $J(S_{t+1})$ two database instances. Then the \cbchase\ calculates a (minimal) sub-database $$I^\ast(S_t) = \chase_{\mathcal{M}^\ast}(J(S_{t+1})) = \chase_{\mathcal{M}^\ast}(\chase_\mathcal{M}(I(S_t)))$$ (highlighted \textcolor{red}{red} in Figure \ref{fig:motivation}). If we are not interested in processing the entire instance $J$, we can alternatively process a subset $J^\ast \subset J$ in the \bchase\ (highlighted in \textcolor{blue}{blue}). For simplicity, we will use $J$ in the following. In summary, we state: Evolution is processed by the \chase\ in the so-called \textit{\chase\ phase}. The \textit{\bchase\ phase} corresponds to the inversion of the evolution.

Since we allow quasi-inverse functions to invert the evolution, we symbolize the inverse mapping by $^\ast$. The exact inverse mapping of $f$ used for merging and splitting, on the other hand, is denoted by $f^{-1}$, which we consider as additional provenance information.

\begin{figure}[ht]
\centering
\includegraphics[width=0.9\linewidth]{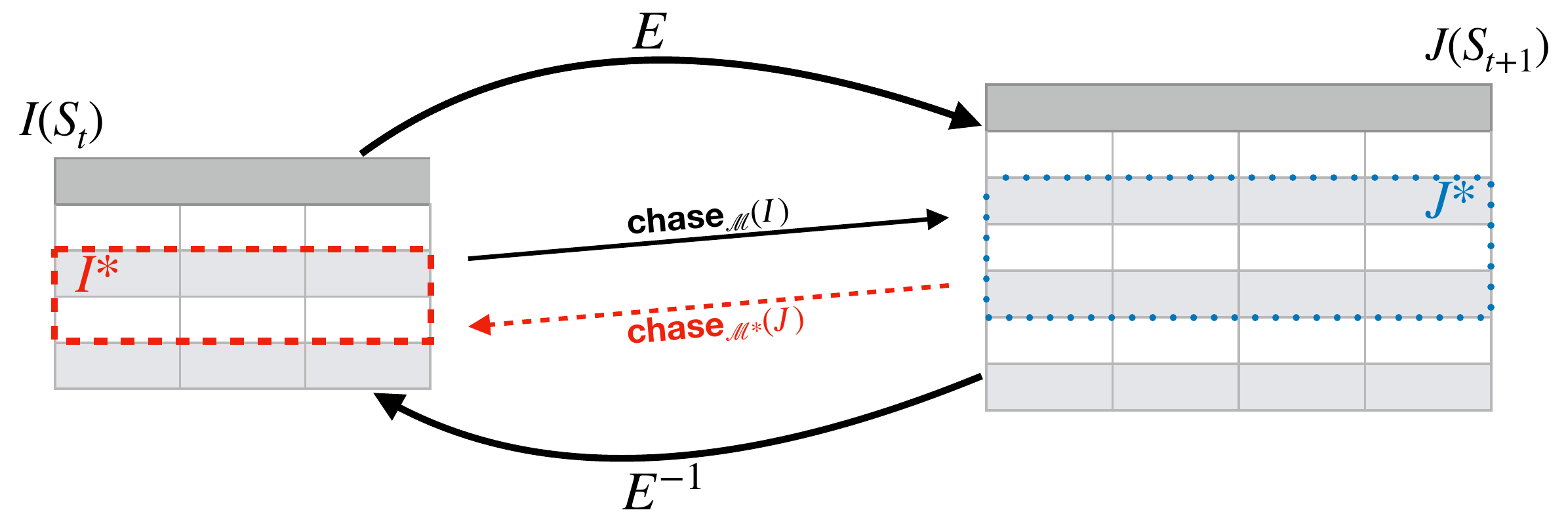}
\caption{Inverting schema evolution}
\label{fig:motivation}
\end{figure}

\paragraph{\textbf{Schema Modification Operators}}
The atomic schema changes during an evolution can be represented by so-called \textit{schema modification operators}. Based on \cite{CMZ08,CMTZ08,AMJFH20,GZ12,WN11,QLS13} and more we define 15 frequently used operators, summarized in Table \ref{tab:SMOinverse}. Some publications also speak of a sixteenth operator \texttt{NOP}, which performs no action but namespace management. 

Each operator has an associated (quasi-)inverse function. The invertibility of an operator corresponds to the existence of a perfect or quasi-inverse and its uniqueness. \texttt{ADD Column}, \texttt{MERGE Column}, \texttt{RENAME Column}, \texttt{SPLIT Column}, \texttt{CREATE Table}, \texttt{PARTITION Table}, \texttt{RENAME Table} and \texttt{NOP} has a unique quasi-inverse function. According to \cite{CMZ08}, \texttt{DROP Table}, \texttt{COPY Table}, \texttt{MERGE Table} and \texttt{DROP Column} have one or more quasi-inverses. Our studies show that the SMO formalization as s-t tgds always has a unique (quasi-)inverse. Therefore, another inverse can exist only in the case of a second s-t tgd formalization. However, this only applies to \texttt{COPY Table} and \texttt{DECOMPOSE Table} (see Table \ref{tab:SMOinverse}).

\begin{table}[ht]
\caption{SMOs with exact ${(=)}$, tp-relaxed $(\preccurlyeq_\textrm{tp})$, relaxed ${(\preccurlyeq)}$ or result equivalent ${(\leftrightarrow)}$ \chase-inverse extended by provenance (highlighted in \textcolor{orange}{orange}); dangling tuples are highlighted in \textcolor{magenta}{magenta} and duplicates in \textcolor{cyan}{cyan}}
\label{tab:SMOInverseExtract}
\setcounter{copyrightbox}{3}
\begingroup
\vfill
\centering
\begin{sideways} 
\scalebox{0.75}[0.9]{
\begin{threeparttable} 
\renewcommand{\arraystretch}{1.25}
\small
\begin{tabular}{|l||c|c|c|c|c|c|}
\hline
& & \multicolumn{2}{c|}{\textbf{Without Provenance}} & \multicolumn{3}{c|}{\textbf{With Provenance}} \\
\cline{2-7}
~\textbf{SMO} & $\Sigma$ & ~\textbf{Type}~ & $\Sigma^{-1}$ & ~\textbf{Type}~ & $\Sigma^{-1}$ & \textbf{Provenance} \\
\hline
\hline
\texttt{COPY Table} & $R(a,b,c) \rightarrow R'(a,b,c) \wedge V(a,b,c)$ & $=$ & $R'(a,b,c) \wedge V(a,b,c) \rightarrow R(a,b,c)$ & $=$ & $R'(a,b,c) \wedge V(a,b,c) \rightarrow R(a,b,c)$ & \\
\cline{2-7}
& $R(a,b,c) \rightarrow R'(a,b,c)$ & $=$ & $R'(a,b,c) \rightarrow R(a,b,c)$ & $=$ & $R'(a,b,c) \rightarrow R(a,b,c)$ & \\
& $R(a,b,c) \rightarrow V(a,b,c)$ & & $V(a,b,c) \rightarrow R(a,b,c)$ & & $V(a,b,c) \rightarrow R(a,b,c)$ & \\
\hline
\texttt{JOIN Table} & $R(a,b) \wedge V(a,c) \rightarrow T(a,b,c)$ & $=$ & $T(a,b,c) \rightarrow R(a,b) \wedge V(a,c)$ & $=$ & $T(a,b,c) \rightarrow R(a,b) \wedge V(a,c)$ & \\
\cline{3-7}
& & \textcolor{magenta}{$\pmb{\preccurlyeq}$} & $T(a,b,c) \rightarrow R(a,b) \wedge V(a,c)$ & $\preccurlyeq$ & $T(a,b,c) \rightarrow R(a,b) \wedge V(a,c)$ & \why, \how \\
\cline{5-7}
& & & & $\textcolor{orange}{=}$ & $T(a,b,c) \rightarrow R(a,b) \wedge V(a,c)$ & \why, \how \\
& & & & & \textcolor{orange}{+ reconstruction of lost tuples} & + side table \\
\hline
\texttt{MERGE Table} & $R(a,b,c) \rightarrow T(a,b,c)$, & $\leftrightarrow$ & $T(a,b,c) \rightarrow R(a,b,c)$, & \textcolor{orange}{$=$} & $T(a,b,c) \rightarrow R(a,b,c)$, & \why, \how \\
 & $V(a,b,c) \rightarrow T(a,b,c)$ & & $T(a,b,c) \rightarrow V(a,b,c)$ & & $T(a,b,c) \rightarrow V(a,b,c)$ & \\
 & & & & & \textcolor{orange}{+ restriction to tuples of $R$ resp. $V$} & \\
\hline \hline
\texttt{MERGE Column} & $R(a,b,c) \rightarrow T(a,f(b,c))$ & $\preccurlyeq_\textrm{tp}$ & $T(a,g) \rightarrow \exists D,E : R(a,D,E)$ & $\preccurlyeq_\textrm{tp}$ & $T(a,g) \rightarrow \exists D,E : R(a,D,E)$ & \\
\cline{3-7}
& & \textcolor{cyan}{$ \pmb{\preccurlyeq}$} & $T(a,b) \rightarrow \exists D,E : R(a,D,E)$ & \textcolor{orange}{$\preccurlyeq_\textrm{tp}$} & $T(a,g) \rightarrow \exists D,E : R(a,D,E)$ & \why, \how\\
& & & & & \textcolor{orange}{+ reconstruction of lost tuples (duplicates)} & \\
\cline{3-7}
& & & & \textcolor{orange}{$=$} & $T(a,g) \rightarrow R(a,f^{-1}(g,c),c)$ & \why, \how \\ 
& & & & & \textcolor{orange}{+ reconstruction of lost attribute values} & + inverse of $f$ \\
& & & & & & + side table\\
\cline{5-7}
& & & & \textcolor{cyan}{$\pmb{=}$} & $T(a,f(b,c)) \rightarrow R(a,f^{-1}(g,c),c)$ & \why, \how \\ 
& & & & & \textcolor{orange}{+ reconstruction of lost tuples (duplicates)} & + inverse of $f$ \\
& & & & & \textcolor{orange}{+ reconstruction of lost attribute values} & + side table \\
\hline
\end{tabular}
\medskip 
\end{threeparttable} 
}
\end{sideways}
\vfill
\endgroup
\end{table}

All operators can be easily formalized as set of one or at most two s-t tgds. Similar to our formalization is for example the formalization of \cite{CMZ08}. However, there is a crucial difference in the interpretation of \texttt{MERGE Table}. The authors define this operator by $R(a,b,c) \rightarrow T(a,b,c)$, $R(a,b,c) \rightarrow V(a,b,c)$ with associated inverse $T(a,b,c) \vee V(a,b,c) \rightarrow R(a,b,c)$. Since \chase-inverses are not defined for \textit{disjunctive embedded dependency}, we consider only s-t tgds. However, the information loss of using a s-t tgd here can be compensated by adding additional provenance information.

\begin{table*}[ht]
\caption{Schema modification operators with corresponding inverse and class}
\label{tab:SMOinverse}
\centering
\scalebox{0.85}{
\begin{tabular}{|l|l|l|c|}
\hline
\textbf{SMO} & \textbf{description} & \textbf{inverse} & \textbf{class} \\
\hline \hline
\texttt{COPY Table} & creates a duplicate of an existing table & \texttt{DROP Table} & I \\
& & \texttt{MERGE Table} & \\
\texttt{CREATE Table} & introduces a new, empty table & \texttt{DROP Table} & I \\
\texttt{DECOMPOSE Table} & decomposes a source table into two tables without changing the data stored inside & \texttt{ADD Column} & III \\
& & \texttt{JOIN Table} & \\
\texttt{DROP Table} & removes an existing table & \texttt{CREATE Table} & IV \\
\texttt{JOIN Table} & combines two tuples from different tables & \texttt{DECOMPOSE Table} & II \\
\texttt{MERGE Table} & takes two source tables and creates a new table that stores their union & \texttt{PARTITION Table} & IV \\
\texttt{PARTITION Table} & distributes tuples into two newly created tables according to the specified condition & \texttt{MERGE Table} & I \\
\texttt{RENAME Table} & changes a table name & \texttt{RENAME Table} & I \\
\hline
\texttt{ADD Column} & introduces a new column with values created by a user-defined constant or function & \texttt{DROP Column} & I \\
\texttt{COPY Column} & copies a column to another table according to the join condition & \texttt{DROP Column} & I \\
\texttt{DROP Column} & removes an existing column from a table & \texttt{ADD Column} & III \\
\texttt{MERGE Column} & sequence of \texttt{ADD Column} and \texttt{DROP Column} using a user-defined function & \texttt{SPLIT Column} & III \\
\texttt{MOVE Column} & same like \texttt{COPY Column}, but the original column is deleted & \texttt{MOVE Column} & II, III \\
\texttt{RENAME Column} & changes the name of a column & \texttt{RENAME Column} & I \\
\texttt{SPLIT Column} & sequence of \texttt{ADD Column} and \texttt{DROP Column} using a user-defined function & \texttt{MERGE Column} & III \\
\hline
\texttt{NOP} & nothing happens & \texttt{NOP} & I \\
\hline
\end{tabular}
}
\end{table*}

Some operators such as \texttt{PARTITION Table}, \texttt{COPY Column} and \texttt{MOVE Column} require special conditions $\textsf{cond}_A$, which give restrictions on the attributes. The corresponding s-t tgds 
are restricted to constant and attribute selection $\sigma_{A_i\theta c}$ or rather $\sigma_{A_i \theta A_j}$ with $\theta \in \{<, \le, =, \ge, >\}$. Other conditions can be formalized using additional side tables.

\textit{Side tables} contain attribute values required for the reconstruction of the (minimal) sub-database defined above. Each tuple is a restriction of a source tuple $t$ and is uniquely identified by its ID $r_t$. This ID is not used as a key, but as additional provenance information. Thus, side tables always exist only in association with a witness basis or a provenance polynomial.

\paragraph{\textbf{Inverse Types}}
Adding Provenance may improve the inverse type. However, in many cases this is not necessary at all. Thus, about half of the operators already guarantee an exact \chase\ inverse without requiring additional information. The associated operators -- \texttt{COPY Table}, \texttt{CREATE Table}, \texttt{PARTITION Table}, \texttt{RENAME Table}, \texttt{ADD Column}, \texttt{COPY Column} and \texttt{RENAME Column} plus \texttt{NOP} group together to form Class I. In the case of the remaining operators, additional provenance causes an improvement of the inverse type.

To cover the scope of \chase-inverse types as well as their provenance extensions, and to cover special cases such as dangling tuples and duplicates, we focus our analyses here on \texttt{COPY Table} (Class I), \texttt{JOIN Table} (Class II), \texttt{MERGE Table} (Class IV) and \texttt{MERGE Column} (Class III). The results are summarized in Table \ref{tab:SMOInverseExtract}. Each operator as well as its associated inverse is formalized as schema mapping $\mathcal{M} = (S_t, S_{t+1}, \Sigma)$ respectively $\mathcal{M}^\ast = (S_{t+1},S_t,\Sigma^{-1})$ with $\Sigma$ (column 2) and $\Sigma^{-1}$ (columns 4 and 6) sets of s-t tgds. The inverse types are summarized in column 3 and 5. The last three columns describe the facts including additional provenance.

Splitting, merging or deleting columns often creates lost tuples. Tuples can also be eliminated when joining tables. Thus, we distinguish two cases: duplicates (Table \ref{tab:SMOInverseExtract}, highlighted in \textcolor{cyan}{cyan}) and dangling tuples (highlighted in \textcolor{magenta}{magenta}). Duplicates are generated by the operators \texttt{DECOMPOSE Table}, \texttt{DROP Column}, \texttt{MERGE Column}, \texttt{MOVE Column} and \texttt{SPIT Column}. They are summarized in Class III. Since \texttt{MOVE Column} may produce dangling tuples as well as duplicates, this operator is also part of Class II, along with \texttt{JOIN Table}. Except for \texttt{DROP Table} and \texttt{PARTITION Table}, all operators are thus assigned to one of the classes 1 to 3. These form class IV, the class of further operators.

As seen in Table \ref{tab:SMOInverseExtract}, the formalization of some operators as set of s-t tgds is not always unique. Thus, \texttt{COPY Table}, \texttt{DECOMPOSE Table}, \texttt{ADD Column} and \texttt{COPY Column} can be represented in two different ways. The choice of the representation influences the resulting \chase-inverse type as well as type of the underlying dependencies. Even additional provenance information may affect the type.

\paragraph{\textbf{Inverse Types using Provenance}}
Reconstructing lost attribute values in \texttt{JOIN Table} provides an exact \chase-inverse instead of a relaxed inverse, whereas the restriction to the relevant tuples in \texttt{MERGE Table} provides an exact \chase-inverse instead of a result-equivalent.
Also, in the case of \texttt{MERGE Column}, additional provenance information can be used to improve the \chase-inverses significantly. Knowing the witness bases or provenance polynomials allows the reconstruction of duplicates. If we also know the inverse of function $f$ and store lost attribute values in a separate side table, it is also possible to specify an exact \chase\ inverse. Thus, provenance enables, among other things, the reconstruction of lost tuples (duplicates) and attribute values.

All in all, we obtain four classes: Class I contains all operators whose inverse type does not change due to additional provenance information. The operators of Class II and Class III deal with dangling tuples or duplicates. Here additional witness basis, provenance polynomials and side tables are indispensable. The remaining two operators are summarized in Class IV. The results of Table \ref{tab:SMOInverseExtract} are described in detail in Section \ref{sec:Examples}. 

\paragraph{\textbf{Composition of several SMOs}}
Since the SMOs operate independently within a sequence, we can determine the inverses in isolation from each other. So, the inverse mapping of composition $$\mathcal{M}^\ast = (\mathcal{M}_1 \circ ... \circ \mathcal{M}_n)^{-1} = \mathcal{M}^\ast_n \circ ... \circ \mathcal{M}^\ast_1$$ results in a composition of the inverse sub-operations $\mathcal{M}^\ast_1, ..., \mathcal{M}^\ast_n$. Thus, the inverse type of $\mathcal{M}^\ast$ corresponds to the type of the weakest partial inverse $\mathcal{M}^\ast_i$. If no inverse exists for one of the partial operations $\mathcal{M}_i$, neither for its composition $\mathcal{M} = \mathcal{M}_1 \circ ... \circ \mathcal{M}_n$.

For example, \texttt{MERGE Column} is defined as a sequence of \texttt{ADD Column} and \texttt{DROP Column} operators. \texttt{DROP Column} has a (tp-)relaxed \chase-inverse depends on whether duplicates exist and additional provenance is used. \texttt{ADD Column}, on the other hand, always has an exact \chase-inverse. So, the inverse type of \texttt{MERGE Column} corresponds to the inverse type of \texttt{DROP Column}. This results in tp-relaxed inverse with additional provenance and a relaxed inverse without. Adding side tables we can define an exact \chase-inverse, such as summarized in Table \ref{tab:SMOInverseExtract}. Same applies for \texttt{SPLIT Column} and other user-defined operators.

\begin{figure*}[ht]
\centering
\begingroup
\centering
\scalebox{0.7}{
\begin{tabular}{cccccccc} 
$R$: & \texttt{id} & \texttt{name} & & $V$: & \texttt{name} & \texttt{subject} & \\
\cline{2-3} \cline{6-7}
& 1 & Alice & \textcolor{orange}{$r_1$} & & Alice & Math & \textcolor{orange}{$s_1$} \\
& \textcolor{magenta}{$\pmb{2}$} & \textcolor{magenta}{$\pmb{Bob}$} & \textcolor{orange}{$r_2$} & & Alice & IT & \textcolor{orange}{$s_2$} \\
\end{tabular}
\xymatrix{
\ar[r]^{\chase_\mathcal{M}} & \quad \\
}
\begin{tabular}{ccccc}
$T$: & \texttt{id} & \texttt{name} & \texttt{subject} & \\
\cline{2-4}
& 1 & Alice & Math & \textcolor{orange}{$r_1 \cdot s_1$} \\
& 1 & Alice & IT & \textcolor{orange}{$r_1 \cdot s_2$} \\
\end{tabular}
\colorbox[gray]{0.95}{ + 
\begin{tabular}{ccc}
& \texttt{id} & \texttt{name} \\
\cline{2-3}
& \textcolor{magenta}{$\pmb{2}$} & \textcolor{magenta}{$\pmb{Bob}$} \\
\end{tabular} 
\begin{tabular}{l}
\\
\textcolor{orange}{$r_2$} \\
\end{tabular}}
\xymatrix {
\ar[r]^{\chase_{\mathcal{M}^\ast}} & \quad \\
}
\begin{tabular}{ccccccc}
$R$: & \texttt{id} & \texttt{name} & ~ & $V$: & \texttt{name} & \texttt{subject} \\
\cline{2-3} \cline{6-7}
& 1 & Alice & & & Alice & Math \\
& \cellcolor[gray]{0.95}{2} & \cellcolor[gray]{0.95}{Bob} & & & Alice & IT \\
& & & & & 
\end{tabular}
}
\endgroup
\caption{\texttt{JOIN Table} with dangling tuple (highlighted in \textcolor{magenta}{magenta}); extended by provenance (highlighted in \textcolor{orange}{orange}) and side table (highlighted in \textcolor{gray}{gray})}
\label{fig:join}
\end{figure*}

\paragraph{\textbf{Similarities to Query Evaluation}}
Processed by the \chase\, schema modification and evaluation queries have similar characteristics. \texttt{MERGE Table}, for example, corresponds to union $\cup$, \texttt{PARTITION Table} to constant or attribute selection $\sigma_{A_i \theta c}$ respectively $\sigma_{A_i \theta A_j}$, and \texttt{JOIN Table} refers to the natural join $\bowtie$. It follows that the inverse types match those of \cite{AH18b}.
In the case of evaluation queries, about half of the inverses can be improved by additional provenance information. The same is true for evolution. In both cases, the (tp)-relaxed \chase-inverse dominates. However, an exact \chase-inverse occurs considerably more frequently in evolution than in evaluation query inversion, see Table \ref{tab:SMOInverseExtract} and \cite{AH18b}.

\section{The four Classes of SMOs}
\label{sec:Examples}

The most common scheme modifications can be summarized in 15 schema modification operators (SMOs). Some operators have an exact inverse function, but the majority does not. Their inverses are quasi-inverses of the types (tp-)relaxed ($\preccurlyeq_\textrm{tp}$) or result equivalent ($\leftrightarrow$). Remedy here is the addition of provenance, which improves this \chase-inverse type in many cases. For this we distinguish four different classes of SMOs: provenance invariant (Class I), dangling tuples (Class II), duplicates (Class III), and special SMOs (Class IV).
We will examine the four classes one by one. For this, we formulate generally valid Lemmas, one for each class. Each class is presented by one representative. 

Let be $R$, $V$ and $T$ three relations over the attributes \texttt{id}, \texttt{name} and \texttt{subject}. The concrete schemas are given directly in the respective proof. Let further be $r_i$ tuple ID of $i$-th tuple, $w_j$ witness basis and $p_j$ provenance polynomial. 
For a source instance $I$, $J = \chase_{\mathcal{M}^\ast}(I)$ is called \textit{evaluation result} and $I^\ast = \chase_{\mathcal{M}^\ast}(J) = \chase_{\mathcal{M}^\ast}(\chase_\mathcal{M}((I))$ \textit{(minimal) sub-database} of $I$.

\paragraph{\textbf{Class I: Provenance Invariant}}
Adding provenance does not always change the \chase-inverse type. This affects about half of the classified evolution operations. Specifically, these include the SMOs \texttt{COPY Table}, \texttt{CREATE Table}, \texttt{PARTITION Table} and \texttt{RENAME Table} as well as \texttt{ADD Column}, \texttt{COPY Column} and \texttt{RENAME Column}. All SMOs of this class have an  exact \chase-invers.

\begin{lemma}
\label{lem:class1}
Let $\mathcal{M}= (S_t, S_{t+1}, \Sigma)$ be a schema mapping of Class I. Then the corresponding inverse function $\mathcal{M}^\ast = (S_{t+1}, S_t, \Sigma^\ast)$ is always an exact or relaxed \chase-inverse independent of additional provenance information.
\end{lemma}

\begin{proof}
W.l.o.g.\,we will prove the statement for \texttt{COPY Table}. This SMO can be formalized in two different ways. It can be represented both as one or as a set of two s-t tgds. Here we consider the first case. Let $\mathcal{M}=(S_t,S_{t+1},\Sigma)$ and $\mathcal{M}^\ast = (S_{t+1},S_t,\Sigma^{-1})$ be two schema mapping with 
\begin{eqnarray*}
\Sigma & = & \{R(a,b,c) \rightarrow R'(a,b,c) \wedge V(a,b,c)\}, \\
\Sigma^{-1} & = & \{R'(a,b,c) \wedge V(a,b,c) \rightarrow R(a,b,c)\}
\end{eqnarray*}
that formalizes \texttt{Copy Column} and its inverse mapping.
W.l.o.g. let $R$ and $V$ each be restricted to three attributes. Let further $I = \{R(x_{i_1}, x_{i_2}, x_{i_3}) \bigm\vert i \in \mathbb{N}\}$ be an arbitrarily chosen source instance. 
Chasing $\mathcal{M}^\ast$ after chasing $\mathcal{M}$ yields:
\begin{eqnarray*}
I^\ast & = & \chase_{\mathcal{M}^{-1}}(\chase_\mathcal{M}(I)) \\
& = & \chase_{\mathcal{M}^{-1}}(\chase_\mathcal{M}(\{R(x_{i_1}, x_{i_2}, x_{i_3}) \bigm\vert i \in \mathbb{N}\})) \\
& = & \chase_{\mathcal{M}^{-1}}(\{R'(x_{i_1}, x_{i_2}, x_{i_3}), V(x_{i_1}, x_{i_2}, x_{i_3}) \bigm\vert i \in \mathbb{N}\}) \\
\end{eqnarray*}
\begin{eqnarray*}
& = & \{R(x_{i_1}, x_{i_2}, x_{i_3}) \bigm\vert i \in \mathbb{N}\} \\
& = & I
\end{eqnarray*}
Thus, we have no loss of information and obtain an exact \chase-inverse. 
Since we can already guarantee an exact inverse, adding provenance, i.e.\,additional witness bases or provenance polynomials over the source IDs, is not necessary here. We do not receive any additional or necessary information from it.
\end{proof}

In addition to \texttt{COPY table}, \texttt{COPY Column} can also be formalized in two different ways. In both cases an exact \chase-inverse can be specified. For all other SMOs of Class I only one formalization exists, witch guarantees an exact \chase\ inverse with and without using additional provenance.

\paragraph{\textbf{Class II: Dangling Tuples}}
Depending on the source instance, dangling tuples may occur. In this case an exact \chase-inverse can not be guaranteed. Also, the addition of witness bases and provenance polynomials may not be sufficient in some cases. However, if we note the lost tuples in so-called \textit{side tables}, we can also specify exact \chase-inverses. This concerns the SMOs \texttt{JOIN table} and \texttt{MOVE Column}.

\begin{lemma}
\label{lem:class2}
Let $\mathcal{M} = (S_t, S_{t+1}, \Sigma)$ be a schema mapping of Class II. Then the corresponding inverse function $\mathcal{M}^\ast = (S_{t+1}, S_t, \Sigma^\ast)$ is exact if we do not have dangling tuples. However, if such tuples are present, the \chase-inverse can be defined as result equivalent or relaxed. By using provenance as well as additional side tables, we can specify an exact \chase\ inverse, too. 
\end{lemma}

We will take a closer look at \texttt{JOIN Table}. For this, let be $\mathcal{M} = (S_t,S_{t+1},\Sigma)$ and $\mathcal{M}^\ast=(S_{t+1},S_t,\Sigma^{-1})$ two schema mappings with
\begin{eqnarray*}
\Sigma & = & \{R(a,b) \wedge V(c,d) \wedge b = d \rightarrow T(a,b,c)\}, \\
\Sigma^{-1} & = & \{T(a,b,c) \rightarrow R(a,b) \wedge V(c,d) \wedge b=d\}.
\end{eqnarray*}
Let further be $R(\texttt{id}, \texttt{name})$, $V(\texttt{name}, \texttt{subject})$ and $T(\texttt{id}, \texttt{name},\texttt{sub-}$ $\texttt{ject})$ three relations with an arbitrarily chosen source instance. Chasing $\mathcal{M}$ after chasing $\mathcal{M}^\ast$ returns
\begin{eqnarray*}
I^\ast & = & \chase_{\mathcal{M}^\ast}(\chase_\mathcal{M}(I)) \\
& = & \chase_{\mathcal{M}^\ast}(\chase_\mathcal{M}(\{R(1, \textrm{Alice}), R(2,\textrm{Bob}), \\
& & \quad V(\textrm{Alice},\textrm{Math}),V(\textrm{Alice},\textrm{IT})\})) \\
& = & \chase_{\mathcal{M}^\ast}(\{T(1,\textrm{Alice},\textrm{Math}),T(1, \textrm{Alice},\textrm{IT})\}) \\
& = & \{R(1, \textrm{Alice}),V(\textrm{Alice},\textrm{Math}),V(\textrm{Alice}, \textrm{IT})\} \\
& \preccurlyeq & I
\end{eqnarray*}

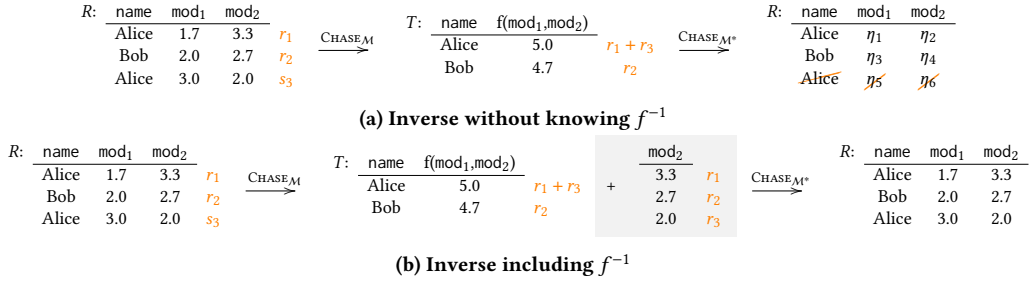
\begin{figure*}[t]
\center
\begin{subfigure}{1\textwidth}
\begingroup
\centering
\scalebox{0.77}{
\begin{tabular}{ccccc} 
$R$: & \texttt{name} & $\texttt{mod}_1$ & $\texttt{mod}_2$ & \\
\cline{2-4}
& Alice & 1.7 & 3.3 & \textcolor{orange}{$r_1$} \\
& Bob & 2.0 & 2.7 & \textcolor{orange}{$r_2$} \\
& Alice & 3.0 & $2.0$ & \textcolor{orange}{$s_3$}
\end{tabular}
\xymatrix{
\ar[r]^{\chase_\mathcal{M}} & \quad \\
}
\begin{tabular}{ccccl}
$T$: & \texttt{name} & f($\texttt{mod}_1$,$\texttt{mod}_2$) & \\
\cline{2-3}
& Alice & 5.0 & \textcolor{orange}{$r_1 + r_3$} \\
& Bob & 4.7 & \textcolor{orange}{$r_2$} \\
\end{tabular}
\xymatrix {
\ar[r]^{\chase_{\mathcal{M}^\ast}} & \quad \\
}
\begin{tabular}{cccccc}
$R$: & \texttt{name} & $\texttt{mod}_1$ & $\texttt{mod}_2$ \\
\cline{2-4}
& Alice & $\eta_1$ & $\eta_2$ \\
& Bob & $\eta_3$ & $\eta_4$ \\
& \colorcancel{Alice}{orange} & \colorcancel{$\eta_5$}{orange} & \colorcancel{$\eta_6$}{orange}
\end{tabular}
}
\caption{Inverse without knowing $f^{-1}$}
\label{fig:MergeCol-1}
\endgroup
\end{subfigure}
\begin{subfigure}{1\textwidth}
\begingroup
\centering
\scalebox{0.77}{
\begin{tabular}{ccccl} 
$R$: & \texttt{name} & $\texttt{mod}_1$ & $\texttt{mod}_2$ & \\
\cline{2-4}
& Alice & 1.7 & 3.3 & \textcolor{orange}{$r_1$} \\
& Bob & 2.0 & 2.7 & \textcolor{orange}{$r_2$} \\
& Alice & 3.0 & 2.0 & \textcolor{orange}{$s_3$}
\end{tabular}
\xymatrix{
\ar[r]^{\chase_\mathcal{M}} & \quad \\
}
\begin{tabular}{cccl}
$T$: & \texttt{name} & f($\texttt{mod}_1$,$\texttt{mod}_2$) & \\
\cline{2-3}
& Alice & 5.0 & \textcolor{orange}{$r_1 + r_3$} \\
& Bob & 4.7 & \textcolor{orange}{$r_2$} \\
\end{tabular}
\colorbox[gray]{0.95}{ + 
\begin{tabular}{cc}
& $\texttt{mod}_2$ \\
\cline{2-2}
& 3.3 \\
& 2.7 \\
& 2.0 \\
\end{tabular} 
\begin{tabular}{l}
\\
\textcolor{orange}{$r_1$} \\
\textcolor{orange}{$r_2$} \\
\textcolor{orange}{$r_3$} \\
\end{tabular}}
\xymatrix {
\ar[r]^{\chase_{\mathcal{M}^\ast}} & \quad \\
}
\begin{tabular}{cccccc}
$R$: & \texttt{name} & $\texttt{mod}_1$ & $\texttt{mod}_2$ \\
\cline{2-4}
& Alice & 1.7 & 3.3 \\
& Bob & 2.0 & 2.7 \\
& Alice & 3.0 & 2.0
\end{tabular}
}
\caption{Inverse including $f^{-1}$}
\label{fig:MergeCol-2}
\endgroup
\end{subfigure}
\caption{\texttt{MERGE Column} using different \chase-inverses $\mathcal{M}^\ast$; extended by provenance (highlighted in \textcolor{orange}{orange}) and side table (highlighted in \textcolor{gray}{gray})}
\label{fig:MergeCol}
\end{figure*}

Like shown in Figure \ref{fig:join}, dangling tuple (highlighted in \textcolor{magenta}{magenta}) can be problematic here. Lost tuples like $(2,\textrm{Bob})$ can not be reconstructed in the (minimal) sub-database $I^\ast$. In this case, we merely reconstruct a subset of the source instance $I$ and record a result-equivalent \chase-inverse. However, if there are no dangling tuples, an exact \chase-inverse can be guaranteed.

Adding witness basiss or provenance polynomials (highlighted in \textcolor{orange}{orange}) does not change the \chase-inverse type. These do not provide any additional information for us. However, if we additionally store the dangling tuple $r_2$ in an external side table (highlighted in \textcolor{gray}{gray}), it can be added back to the reconstructed instance $I^\ast$. Summarized we can guarantee an exact \chase-inverse, because:
\begin{eqnarray*}
I^\ast & = & \{R(1, \textrm{Alice}), S(\textrm{Alice}, \textrm{Math}), V(\textrm{Alice}, \textrm{IT})\} \cup \{R(2, \textrm{Bob})\} \\
& = & I.
\end{eqnarray*}

\texttt{MOVE Column} may contain duplicates in addition to dangling tuple. The corresponding SMO can be formulated in two different ways. It is the only SMO in our study that does not have a \chase-inverse in every case.

\paragraph{\textbf{Class III: Duplicates}}
The presence of duplicates typically permits only the specification of a relaxed \chase-inverse. SMOs belonging to this class are \texttt{DROP Column}, \texttt{MERGE Column}, \texttt{SPLIT Column} and \texttt{DECOMPOSE Table}. An improvement of the inverse type is only possible by adding provenance information. In some cases even an exact \chase-inverse can be guaranteed.

\begin{figure*}[ht]
\centering
\scalebox{0.675}{
\begin{tabular}{ccccccccccc} 
$R$: & \texttt{id} & \texttt{name} & \texttt{subject} & & & $V$: & \texttt{id} & \texttt{name} & \texttt{subject} & \\
\cline{2-4} \cline{8-10}
& 1 & Alice & Math & \textcolor{orange}{$r_1$} & & & 2 & Bob & IT & \textcolor{orange}{$s_1$} \\
& 1 & Alice & IT & \textcolor{orange}{$r_2$} & & & 1 & Alice & IT & \textcolor{orange}{$s_2$} \\
\end{tabular}
\xymatrix{
\ar[r]^{\chase_\mathcal{M}} & \quad \\
}
\begin{tabular}{ccccl}
$T$: & \texttt{id} & \texttt{name} & \texttt{subject} & \\
\cline{2-4}
& 1 & Alice & Math & \textcolor{orange}{$r_1 + s_1$}\\
& 2 & Bob & IT & \textcolor{orange}{$r_2$}\\
& 1 & Alice & IT & \textcolor{orange}{$s_2$}
\end{tabular}
\xymatrix {
\ar[r]^{\chase_{\mathcal{M}^\ast}} & \quad \\
}
\begin{tabular}{ccccccccc}
$R$: & \texttt{id} & \texttt{name} & \texttt{subject} & & $V$: & \texttt{id} & \texttt{name} & \texttt{subject} \\
\cline{2-4} \cline{7-9}
& 1 & Alice & Math & & & 1 & Alice & Math \\
& \colorcancel{2}{orange} & \colorcancel{Bob}{orange} & \colorcancel{IT}{orange} & & & 2 & Bob & IT \\
& 1 & Alice & IT & & & \colorcancel{1}{orange} & \colorcancel{Alice}{orange} & \colorcancel{IT}{orange} \\
\end{tabular}
}
\caption{\texttt{MERGE Table} extended by provenance (highlighted in \textcolor{orange}{orange})}
\label{fig:MergeTab}
\end{figure*}

\begin{lemma}
Let $\mathcal{M} = (S_t, S_{t+1}, \Sigma)$ be a schema mapping of Class III. Then the corresponding \chase-inverse $\mathcal{M}^\ast = (S_{t+1}, S_t, \Sigma^{-1})$ is (tp-)relaxed, depending on the existence of duplicates. By adding provenance information and side tables, the \chase-inverse type can be defined as tp-relaxed or exact.
\end{lemma}

\begin{proof}
Let $\mathcal{M} = (S_t, S_{t+1}, \Sigma)$ and $\mathcal{M}^\ast = (S_{t+1},S_t,\Sigma^{-1})$ be two schema mappings with
\begin{eqnarray*}
\Sigma & = & \{R(a,b,c) \rightarrow T(a,f(b,c))\}, \\
\Sigma^{-1} & = & \{T(a,g) \rightarrow \exists D, E: R(a, D, E)\}
\end{eqnarray*}
formalizing \texttt{MERGE Column} and its associated inverse mapping. W.l.o.g. let $R$ be restricted to three and $T$ to two attributes. Let further $f$ be any MERGE function, $I = \{R(x_{i_1},x_{i_2},x_{i_3})_{r_i} \bigm\vert i \in \mathbb{N}$ $\wedge \ f(x_{j_2},x_{j_3}) = f(x_{k_2},x_{k_3}) \wedge \ r_i \textrm{ is tuple identifier}\}$ an instance with tuple identifier $r_i$ and duplicate in $t_j$ and $t_k$. Let $w_j$ be witness basis of duplicate $r_j$ with $w_j = w_k = \{\{r_j\},\{r_k\}\}$. So, $w_j$ contains two sets of tuple identifiers, so-called \textit{witnesses}. A witness, here $\{r_j\}$ and $\{r_k\}$, holds all tuple IDs needed for reconstructing one tuple. Two witnesses represent a duplicate.
Chasing $\mathcal{M}^\ast$ after chasing $\mathcal{M}$ then returns  
\begin{eqnarray*}
I^\ast & = & \chase_{\mathcal{M}^\ast}(\chase_\mathcal{M}(I)) \\
& = & \chase_{\mathcal{M}^\ast}(\{T(x_{i_1}, f(x_{i_2}, x_{i_3}))_{w_i} \bigm\vert i \in \mathbb{N} \\
& & \quad \wedge \ f(x_{j_2},x_{j_3}) = f(x_{k_2}, x_{k_3})\} \\
& & \quad \wedge \ w_i \textrm{ witness basis with } w_j = w_k = \{\{r_j\},\{r_k\}\}) \\
& = & \{R(x_{i_1}, \eta_{i_2}, \eta_{i_3}) \bigm\vert i \in \mathbb{N} \wedge \ \eta_{j_2} = \eta_{k_2}, \eta_{j_3} = \eta_{k_3}\} \\
& \preccurlyeq_\textrm{(tp)} & I. 
\end{eqnarray*}
Due to the existence quantifiers, the inverse s-t tgd always generates two new null values $\eta_{i_2}$ and $\eta_{i_3}$. Thus, the attribute values processed in $f$ can not be recovered. Due to the duplicate $t_j = t_k$, the reconstructed instance $I^\ast$ contains less tuples than the original instance $I$. The inverse $\mathcal{M}^\ast$ is thus relaxed without and tp-relaxed with existing duplicates. Adding provenance guarantees an tp-relaxed \chase-inverse.  In this case, $\eta_{j_2}\neq\eta_{j_3}$ as well as $\eta_{k_2}\neq\eta_{k_3}$ holds.

Let alternatively be $\mathcal{M}^\ast = (S_{t+1},S_t,\Sigma^{-1})$ inverse mapping with
\allowdisplaybreaks
\begin{eqnarray*}
\Sigma^{-1} & = & \{T(a,b)) \rightarrow R(a,f^{-1}(a,f^{-1}(g,c),c\}.
\end{eqnarray*}
This formalization can only be used in combination with provenance, since polynomials, side tables and the knowledge of $f^{-1}$ are necessary here. Then chasing $\mathcal{M}^\ast$ after chasing $\mathcal{M}$ returns
\begin{eqnarray*}
I^\ast & = & \chase_{\mathcal{M}^\ast}(\chase_\mathcal{M}(I)) \\
& = & \chase_{\mathcal{M}^\ast}(\{T(x_{i_1}, f(x_{i_2},x_{i_3}))_{p_i} \bigm\vert i \in \mathbb{N} \\
& & \quad \wedge \ f \textrm{ invertible function} \\
& & \quad \wedge \ p_i \textrm{ provenance polynomial with } p_j = p_k = r_k + r_k\} \\
& & \quad \wedge \ f(x_{j_2}, x_{j_3}) = f(x_{k_2},x_{k_3})\} \\
& & \cup \ \{(x_{i_3})_{r_i} \bigm\vert i \in \mathbb{N} \wedge \ r_i \ \textrm{tuple identifier}\}) \\
& = & \{R(x_{i_1}, f^{-1}(f(x_{i_2},x_{i_3}),x_{i_1})_{\mid_{x_{i_2}}},f^{-1}(f(x_{i_2},x_{i_3}),x_{i_1})_{\mid_{x_{i_3}}}) \bigm\vert \\
& & \quad i \in \mathbb{N} \wedge \ f(x_{j_2}, x_{j_3}) = f(x_{k_2},x_{k_3})\} \\
& = & \{R(x_{i_1}, x_{x_2}, x_{i_3}) \bigm\vert i \in \mathbb{N}\} \\
& = & I.
\end{eqnarray*}
Let be $p_i$ provenance polynomial and $a_{i_j}$ element of the side table $\{(x_{i_3})_{r_i} \bigm\vert i \in \mathbb{N} \wedge \ r_i \ \textrm{tuple identifier}\}$. Then $x_{i_2}$ can be recalculated from $f^{-1}(g,x_{i_3})$ as follows: 
\begin{itemize}
\vspace{-0.1cm}
\item attribute values $g$ and $x_{i_3}$ are explicitly contained in $J = \chase_\mathcal{M}(I)$;
\item inverse function $f^{-1}$ is given after precondition;
\item attribute value $x_{i_{3}}$ can be read from the side table using $p_i$.
\end{itemize}
Also duplicates can be reconstructed using this additional information. Overall an exact \chase-inverse can be guaranteed.
\end{proof}

Figure \ref{fig:MergeCol} shows the situation above in a concrete example. Let be $R$ a student-table where the points of the two passed modules should be summarized. The inverse without knowing $f^{-1}$ yields a relaxed \chase-inverse (see Figure \ref{fig:MergeCol-1}). However, the duplicate $(\textrm{Alice},5.0)$ in $T$ can be reconstructed using additional provenance information, here a provenance polynomial (highlighted in \textcolor{orange}{orange}). Adding a sustainable side even allows the definition of a exact \chase-inverse if $f^{-1}$ is known (see Figure \ref{fig:MergeCol-2}). Proving the other SMOs of Class III is quite similar.

\paragraph{\textbf{Class IV: Special SMOs}}
Almost all SMOs can be assigned to Classes 1 to 3. Only \texttt{MERGE Table} and \texttt{DROP Table} remain. These contain neither dangling tuple (see Class II) nor duplicates (see Class III). In contrast to Class I, however, adding provenance improves the inverse type in both cases.

\begin{lemma}
\label{lem:class41}
The SMO \texttt{MERGE Table} has a result equivalent \chase-inverse. Adding provenance information, the \chase-inverse type can be defined as exact.
\end{lemma}

The unification of two tables $R$ and $V$ can be formalized as schema mapping $\mathcal{M}$ with inverse mapping $\mathcal{M}^\ast = (S_{t+1},S_t,\Sigma^{-1})$,
\begin{eqnarray*}
\Sigma & = & \{R(a,b,c) \rightarrow T(a,b,c),V(a,b,c) \rightarrow T(a,b,c)\}, \\
\Sigma^{-1} & = & \{T(a,b,c) \rightarrow R(a,b,c), T(a,b,c) \rightarrow V(a,b,c)\}.
\end{eqnarray*}
Let further be $R$, $S$ and $T$ three relations with an arbitrarily chosen source instance $I$. Chasing $\mathcal{M}^\ast$ after chasing $\mathcal{M}$ then yields
\begin{eqnarray*}
I^\ast & = & \chase_{\mathcal{M}^\ast}(\chase_\mathcal{M}(I)) = \chase_{\mathcal{M}^\ast}(J) \neq I,
\end{eqnarray*}
whereby $I^\ast$ may contain more tuples than $I$, since $\mathcal{M}^\ast$ preserves all existing tuples like $(2,\textrm{Bob},\textrm{IT})$ and $(1,\textrm{Alice},\textrm{IT})$ in Figure \ref{fig:MergeTab}, which are not originally contained in $R$ and $T$, respectively.

However, the redundant tuples in $I^\ast$ can be identified and deleted by using \how-provenance (highlighted \textcolor{orange}{orange} in Figure \ref{fig:MergeTab}). Knowing the tuples used for calculating $J$, allows us to reassign the tuples to their source relations. We can thus reconstruct the source instance error-free and guarantee an exact \chase-inverse. It is not significant how exactly the result is produced. It is sufficient to know which tuple ids are involved in the generation of $(2, \textrm{Bob}, \textrm{IT})$ or $(1, \textrm{Alice}, \textrm{IT})$. A generalized witness basis or witness list is sufficient, too.

\begin{lemma}
\label{lem:class42}
The SMO \texttt{DROP Table} has a relaxed \chase-inverse. Adding provenance information the \chase-inverse type can be defined as tp-relaxed.
\end{lemma}


\section*{Conclusion}
Especially for long-term data-driven studies, changing databases lead at some point to a large number of schemas. However, in the sense of reproducibility of research data each database version must be reconstructable with little effort. Nevertheless, in many cases, such an evolution can not be fully reconstructed.

This article classified the 15 most frequently used schema modification operators and defined the associated inverses for each operation. For avoiding an information loss, it furthermore defined which additional provenance information have to be stored. All in all, we define four classes dealing with dangling tuples, duplicates and provenance-invariant operators.

By using and extending the theory of schema mappings and their inverses for queries, we are able to combine data analysis applications with provenance under evolving database structures, thus enabling the reproducibility of scientific results over longer periods of time.

\end{document}